\documentclass[11pt]{article}

\usepackage{graphicx} 
\usepackage[T1]{fontenc}
\usepackage[utf8x]{inputenc}
\usepackage{amsmath, amssymb, amsthm}
\usepackage{csquotes}
\usepackage[english]{babel}
\usepackage{eurosym}
\usepackage{fullpage}
\usepackage{palatino}
\usepackage{tikz}
\usepackage{verbatim}
\usepackage{fancybox}

\usepackage{subfig}
\usepackage{enumerate}
\usepackage{epsfig}
\usepackage[active]{srcltx}
\usepackage{amsfonts}
\usepackage{amsmath}
\usepackage[mathlines]{lineno}
\usepackage{xspace}
\usepackage{tikz}
\usepackage{authblk}

\usetikzlibrary{arrows,shapes,snakes,automata,backgrounds,petri,patterns}

\newtheorem{lemma}{Lemma}

\newtheorem{theorem}{Theorem}
\newtheorem{claim}{Claim}

\newtheorem{question}{Question}

\newcounter{claimb}
    
\def\claimb{$$\vcenter\bgroup\advance\hsize by -8em\noindent
\refstepcounter{claimb}\ignorespaces\it}        
\makeatletter
\def\endclaimb{\rm\egroup\leqno(\theclaimb)$$\global\@ignoretrue}
\makeatother

    {\noindent \emph{Proof.} {}{#1}{}}{\hfill
    $\Diamond$\vspace{1em}}

\begin{document}

\title{Reconfiguring Independent Sets in Cographs\thanks{Work partially supported by the ANR Grant EGOS (2012-2015) 12 JS02 002 01}}


\author[1]{Marthe Bonamy \thanks{marthe.bonamy@lirmm.fr}}
\author[1,2]{Nicolas Bousquet\thanks{nicolas.bousquet@lirmm.fr}}
\affil[1]{LIRMM, Universit\'e Montpellier 2, France}
\affil[2]{Department of Mathematics and Statistics, McGill University, and GERAD, Montr\'eal}

\renewcommand\Authands{ and }

\date{}
\maketitle

\begin{abstract}
Two independent sets of a graph are \emph{adjacent} if they differ on exactly one vertex (\emph{i.e.} we can transform one into the other by adding or deleting a vertex). 
Let $k$ be an integer. We consider the reconfiguration graph $TAR_k(G)$ on the set of independent sets of size at least $k$ in a graph $G$, with the above notion of adjacency. Here we provide a cubic-time algorithm to decide whether $TAR_k(G)$ is connected when $G$ is a cograph, thus solving an open question of~[Bonsma 2014]. As a by-product, we also describe a linear-time algorithm which decides whether two elements of $TAR_k(G)$ are in the same connected component.\vspace{5pt} 
\end{abstract}

\section{Introduction}
Reconfiguration problems (see~\cite{BonamyB14,Gopalan09,Ito2011,Ito2009} for instance) consist in finding step-by-step transformations between two feasible solutions such that all intermediate results are also feasible. Such problems model dynamic situations where a given solution is in place and has to be modified, but no property disruption can be afforded. Two types of questions are interesting concerning reconfiguration problems: in which case can we ensure that there exist such a transformation? And what is the complexity of finding such a reconfiguration? In this paper we concentrate on algorithmic aspects of reconfiguration of independent sets. Reconfiguration of independent sets naturally appears when we deal with the reconfiguration of geometric objects (see~\cite{HearnD05}) and received a lot of attention in the last few years (e.g.~\cite{b14,BonsmaKW14,HearnD05,KaminskiMM12}).

In the whole paper, $G=(V,E)$ is a graph where $n$ denotes the size of $V$, and $k$ is an integer. For standard definitions and notations on graphs, we refer the reader to~\cite{Diestel2005}.
A \emph{$k$-independent set} of $G$ is a set $S \subseteq V$ with $|S|\geq k$, such that no two elements of $S$ are adjacent. Two $k$-independent sets $I, J$ are \emph{adjacent} if they differ on exactly one vertex (\emph{i.e.} if there exists a vertex $u$ such that $I = J \cup \{u\}$ or the other way round). This model is called the \emph{Token Addition and Removal} (TAR). Two other models using other notions of adjacencies have been introduced. In the token jumping (TJ) model~\cite{Ito2011}, two independent sets are adjacent if one can be obtained from the other by replacing a vertex by another one (in particular it means that we only look at independent sets of a given size). In the token sliding (TS) model~\cite{HearnD05}, tokens can be moved along edges of the graph, \emph{i.e} vertices can only be replaced by vertices which are adjacent to them (see~\cite{b14} for an general overview of the results for all these models). 

In this paper we only deal with the TAR model. The \emph{$k$-TAR-reconfiguration graph of $G$}, denoted $TAR_k(G)$, is the graph whose vertices are $k$-independent sets of $G$, with the adjacency defined above. Given a graph $G$, deciding if two $k$-independent sets are in the same connected component in $TAR_k(G)$ is $\mathbf{PSPACE}$-complete even in the case of perfect graphs or of planar graphs with maximum degree three~\cite{HearnD05,KaminskiMM12}. Moreover, deciding if there exists a reconfiguration between two $k$-independent sets with at most $\ell$ operations is strongly NP-complete~\cite{KaminskiMM12}.

Nevertheless, on some more restrictive classes of graphs, there exist positive results, for instance for claw-free graphs~\cite{BonsmaKW14} or line graphs~\cite{Ito2011}. In the class of cographs, Kaminski et al.~\cite{KaminskiMM12} proved that there exists an efficient algorithm for the reachability problem in the TS model. Bonsma recently proved in~\cite{b14} that, in the TAR model, this decision problem is polynomial (even quadratic) for the class of cographs, \emph{i.e.} the class of graphs with no induced path on four vertices. The following question was then asked.

\begin{question}[Bonsma~\cite{b14}]
What is the complexity of deciding whether $TAR_k(G)$ is connected, for $G$ a cograph?
\end{question}

We solve it here by arguing that there is a cubic time algorithm to decide whether $TAR_k(G)$ is connected. The algorithm and its complexity will be detailed in Section~\ref{sec:connected}. In Section~\ref{sec:linear}, we show that a similar argument can also be used to decide in linear-time whether two independent sets belong to the same connected component of $TAR_k(G)$, thus improving the result of Bonsma~\cite{b14}. Our proof is constructive in the sense that our linear time algorithm provides a sequence of operations which transform the first independent set into the second one.

\section{Deciding the connectivity of $TAR_k(G)$}\label{sec:connected}

A graph $G=(V,E)$ is a \emph{cograph}~\cite{lerchs71} if it does not contain induced paths of length $4$. Equivalently, a graph is a cograph if:
\begin{itemize}
 \item $G$ is a single vertex.
 \item Or $V$ can be partitioned into $V_1,V_2$ such that $G[V_1]$ and $G[V_2]$ are cographs and there is no edge between any vertex of $V_1$ and any vertex of $V_2$.
 \item Or $V$ can be partitioned into $V_1,V_2$ such that $G[V_1]$ and $G[V_2]$ are cographs and every vertex of $V_1$ is adjacent to every vertex of $V_2$ (such an operation is called a \emph{join}).
\end{itemize}

Consequently, each cograph admits a so-called cotree decomposition, which is a rooted binary tree whose leaves are all vertices in the graph, and whose nodes are either "disjoint union" or "join" operations (see Figure~\ref{fig:tree}). Note that a cotree has $(n-1)$ internal nodes since it has $n$ leaves which are the vertices of $G$. Given a node $B$ of the cotree, \emph{the cograph induced by $B$} is the subgraph induced in $G$ by the vertices that correspond to a leaf that $B$ cuts from the root. By abuse of notation, $\alpha(B)$ refers to the size of the largest independent set in the cograph induced by $B$.

\begin{theorem}\label{thm:cographs}
Given a cograph $G$ and an integer $k$, it can be decided in cubic time (in $n$) whether $TAR_k(G)$ is connected or not.
\end{theorem}

\begin{proof}
Let $G$ be a cograph and $k$ be an integer. 

\begin{figure}
\centering
 \includegraphics[scale=0.6]{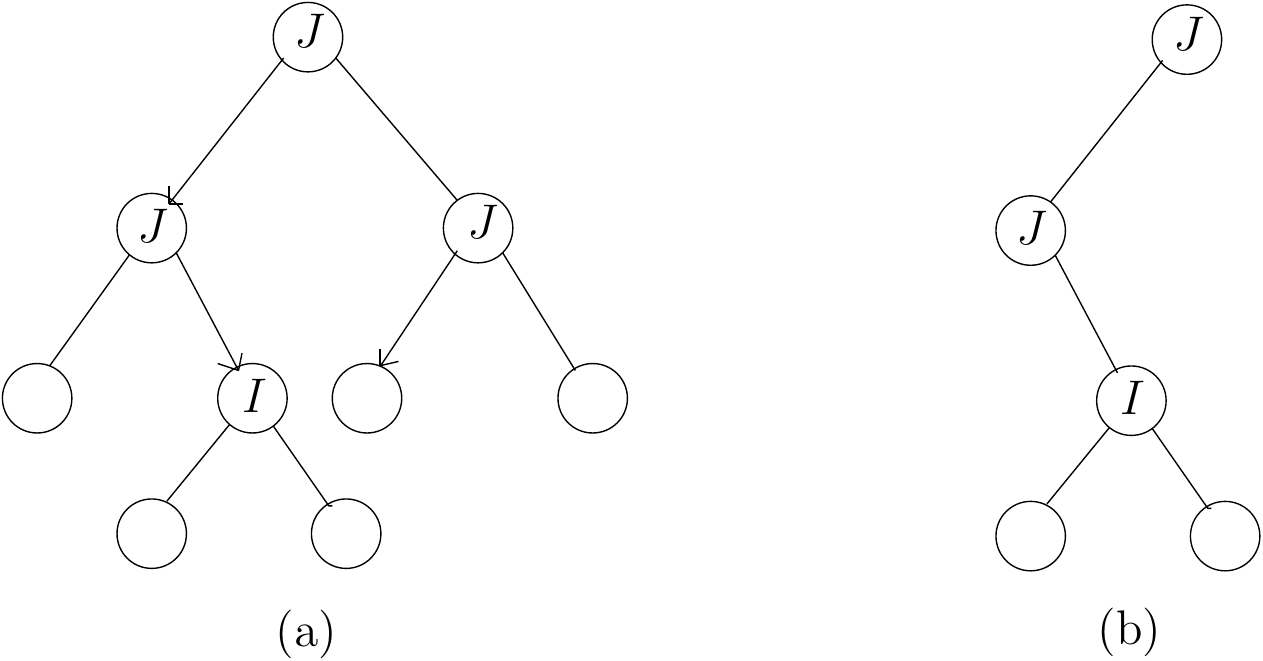}
 \caption{(a) A cotree decomposition where join nodes have been oriented. The leaves correspond to the vertices of $G$. Arcs represent good sides.
          (b) The stable-search corresponding to the maximal independent set where we always make good choices.}
\label{fig:tree}
\end{figure}

We first compute a cotree decomposition of $G$. This can be done in linear time~\cite{Corneil85}. Now, starting bottom-up from the leaves, we compute for every node $B$ the size of the largest independent set in the cograph induced by $B$, as follows. 
\begin{itemize}
 \item If $B$ is a leaf, the answer is $1$ (the unique vertex is an independent set). 
 \item If $B$ is a join node, the answer is the maximum of the size of the largest independent sets in its two sons. 
 \item If $B$ is a disjoint union node, the answer is the sum of the size of the largest independent sets in its two sons. 
\end{itemize}
Let us show that it can be done in linear time.
\begin{claim}\label{claim:nodeconstant}
 An array containing the values of the maximum independent sets contained in the cograph induced by $B$ (for every $B$) can be computed in linear time.
\end{claim}
\begin{proof}
We say that a node is \emph{treated} if we have already determined its value. The only fact that we have to prove is that we can find a node which is not yet treated but whose sons are treated in constant time. We initialize a list containing all the leaves of the cotree. For every internal node, we create a variable which will contain the number of ``treated'' sons. We initialize this variable to $0$. We also create an array which will contain the size of the maximum independent set in the cograph induced by $B$. \\
As long as the list is not empty, we ``treat'' the head of the list and then delete it from the list. If it is a leaf, we give it value one in the array; if it is a join node, we give it the maximum value of its sons; if it is a disjoint union node, we give it the sum of the values of its sons. In any case we increase by one the number of treated sons of the father of $B$. If the number of treated sons of the father equals two, we put the father on the list.
\end{proof}

We consider the set of independent sets that are maximal for inclusion. We say that an independent set is \emph{maximal} if it is maximal by inclusion.
The above algorithm can be transformed into a quadratic time algorithm computing all the possible sizes of maximal independent sets of $G$. Even if the algorithm is simple, its complexity analysis is a little bit more involved.

\begin{claim}\label{claim:precomp}
We can compute in quadratic time the lists of all possible sizes of maximal independent sets of all nodes $B$ of the cotree of $H$.
\end{claim}
\begin{proof}
For each node $B$, we create a list containing all possible sizes of maximal independent sets of the cograph induced by $B$. Now, starting bottom-up from the leaves, we compute for every node $B$ this list, as follows:
\begin{itemize}
 \item If $B$ is a leaf, the list only contains the integer one. 
 \item If $B$ is a join node, the list is the disjoint union of the lists of its sons. 
 \item If $B$ is a disjoint union node, the list contains all the possible sums of one element on the list of a son and one on the other. 
\end{itemize}
As in Claim~\ref{claim:nodeconstant}, a node to treat can be found in constant time. Moreover, a list $L$ which contains values between $1$ and $n$ can be sorted in time $\mathcal{O}(\max(n,|L|))$. This is a classical algorithmic result: we create a binary array of size $n$ initialized to $0$. Then we read the values of the list and put the entry corresponding to each value to $1$. Finally we just extract a new list from this array by putting a value in the list if its corresponding entry has value one in the array.

The computation of the list on a leaf takes a constant time. For a join node, it takes a linear time (since each list contains at most $n$ values). Things are more technical for disjoint union node. Indeed, since we have to compute all the possible sums of values, it can take up to a quadratic time. And since there could be up to $n$ join nodes, the global complexity can be cubic. Nevertheless, we can do better than this naive analysis. Consider a node $B$. Assume that the cograph induced by $B$ has $k$ vertices, with $k \geq 2$. Let $B_1$ and $B_2$ be the cographs induced by each of the two sons of $B$, and let $k_1,k_2$ be their respective number of vertices in the cographs. Note that we have $k_1+k_2=k$. The complexity $C(B)$ for finding the value on this node is:
\[C(B) = \left\{\begin{array}{ll}
 C(B_1)+C(B_2) + \max(k_1,k_2) & \mbox{if $B$ is a disjoint union node}\\
 C(B_1)+C(B_2) + k_1\cdot k_2 & \mbox{if $B$ is a join node}
 \end{array}\right.\]
 
By abuse of notation, we set $C(k)$ to be the maximum of $C(D)$ for a cograph $D$ on $k$ vertices. Then we obtain that the function $C$ satisfies $C(1)=1$ and, for $k \geq 2$, the following equation:

\[C(k) \leq max\{C(k_1)+C(k_2) + k_1\cdot k_2| k_1,k_2 \in \mathbb{N}^*, \ k=k_1+k_2\}.\]
 
Note indeed that $k_1 \cdot k_2 \geq max(k_1,k_2)$.
In particular, the function $x\mapsto x^2$ satisfies these equations. So we can show by recurrence that the complexity of the algorithm is at most quadratic in~$n$.
\end{proof}

For every join node $B$ with sons $B_1$ and $B_2$, we define a \emph{good} side and a \emph{bad} side. If $\alpha(B_1) \neq \alpha(B_2)$, the side with the largest independent set is the good side. If $\alpha(B_1) = \alpha(B_2)$, we choose arbitrarily a good side (see Figure~\ref{fig:tree}(a)).

A \emph{stable-search} is a subtree $T$ of the cotree of $G$ containing the root and such that for every disjoint union node $B$ of $T$ both sons of $B$ are in $T$, and for every join node $J$ of $T$ exactly one son of $J$ is in $T$. Intuitively, it corresponds to ``making no choice'' when we arrive on a disjoint union node and ``choosing exactly one son'' when we arrive on a join node. By extension, \emph{making a bad choice} at a join node is going for the bad son, and similarly for ``good''. In Figure~\ref{fig:tree}(b), the stable-search represents the stable-search where we always make good choices.

\begin{claim}
 There is a bijection between the maximal independent sets and the stable-searches.
\end{claim}
\begin{proof}
Let $I$ be an independent set. The subtree $T$ of $I$ is the subtree of the cotree of the graph defined inductively as follows:
\begin{itemize}
 \item The leaves corresponding to the vertices of $I$ are in $T$.
 \item If a node is in $T$, its father is in $T$.
\end{itemize}
Obviously, the root is in the subtree of $I$. A maximal independent set appears in precisely one side of a join node, so the subtree of $S$ only contains vertices on one side of a join node. And if a maximal independent set appears on one side of a disjoint union node, it should also appear on the other side for the sake of maximality. Therefore, every maximal independent set corresponds to a stable-search, and every stable-search builds a unique maximal independent set, hence the claim. 
\end{proof}
In the following, our goal is to prove that we can transform any maximal independent set of size at least $k$ into the maximal independent set which always follows the good side (which happens to be a maximum independent set), or provide an example of an independent set of size at least $k$ which cannot be transformed into it. To do so, we repeatedly prune our cotree to get rid of one join node at a time while ensuring that we thus obtain an equivalent instance for the problem. By choosing carefully the join node we consider for pruning, we can argue that if it cannot be pruned without possibly obtaining a non-equivalent instance, then the reconfiguration graph is not connected.

We start with $G_0=G$ and build at each step a new cograph $G_{i+1}$ that is obtained from $G_i$ by deleting the subgraph induced by a son of a join node. After such a deletion, a node has only one son, so we contract it with its father.

At each step $i$, until there is is no join node anymore (which means that all the remaining nodes in the cotree are disjoint union nodes, \emph{i.e.} the remaining graph is a stable set), we consider the bottom-most join nodes (by a bottom-most join node, we mean a join node such that the subgraph induced by any of its sons is a stable set, \emph{i.e.} only contains disjoint union nodes), and take one which minimized $\alpha(B)$, where $B$ is the cograph induced by its bad side. Note that such a node can easily be found in linear time. 

Let $\mathcal{L}$ be the list of possible sizes of maximal stable sets in the cograph $G \setminus (B \cup N(B))$, as obtained from Claim~\ref{claim:precomp}. Let us show that deciding if there is a solution or not can be deduced from the values in $\mathcal{L}$.
\begin{lemma}
 \begin{itemize}
  \item If $\mathcal{L}$ contains a value between $k-1$ and $k-\alpha(B)$, then the answer is NO.
  \item If no value of$\mathcal{L}$ is between $k-1$ and $k-\alpha(B)$, then for $G_{i+1}$ the graph $G_i$ where the subgraph induced by $B$ has been deleted, $TAR_k(G_i)$ is connected if and only if $TAR_k(G_{i+1})$ is. 
 \end{itemize}
\end{lemma}
\begin{proof}
First assume that there is a possible value $\beta$ between $k-1$ and $k-\alpha(B)$. Let $I$ be a maximal independent set of $G \setminus (B \cup N(B))$ of size $\beta$ and $I_B$ be an independent set of $B$ of size exactly $\alpha(B)$. Note that $|I \cup I_B| \geq k$. By minimality of $B$, the number of vertices of $I \cup I_B$ on each ``bad side'' is at least $\alpha(B)$. Moreover, we cannot add a vertex in some good side if there remains a vertex of the independent set in its corresponding bad side (remember good and bad sides are joined). 
So in $I \cup I_B$, if we want to add a vertex in a good side, we have to delete at least $\alpha(B)$ vertices, which is impossible since we want to keep an independent set of size at least $k$ (and it is impossible to ``gain'' vertices at other places since the solution is maximal everywhere except on bad sides).

Let us formalize this intuition. Let $J$ be a maximal stable set. An \emph{excellent node for $J$} is a node $B$ such that $J\cap B \neq \varnothing$ and $J$ only makes good choices on the sub-cotree rooted in $B$. For instance, if the root is an excellent node for $J$, then the maximal stable set $J$ is maximum. Let us partition the vertex set into two parts. A \emph{reachable vertex from $J$} is it is in the cotree induced by an excellent node of $J$. A \emph{target vertex of $J$} is a vertex which is not reachable.

\begin{claim}\label{claim:target}
 No independent set of the connected component of $TAR_k(G_i)$ containing $I$ contains a target vertex of $I$.
\end{claim}
\begin{proof}
Assume by contradiction that there exists a sequence of additions and deletions of vertices starting from $I$ which leads to an independent set $J$ containing a target vertex of $I$, and take $J$ to be the first independent set in that sequence containg a target vertex of $I$. Denote by $x$ the target vertex and by $J'$ the set $J \setminus x$. The set $J'$ only contains reachable vertices. Note that $J'$ is just before $J$ in the sequence.

Let us first show that we can assume that $J'$ only contains vertices of $I$. Let $y$ be a vertex of $J'$. Since $y$ is reachable, it means that there exists an ancestor $N_y$ of $y$ such that $B_y$ is an excellent node. W.l.o.g. $N_y$ is the first such ancestor. Assume that $y \notin I$. It means, since $N_y$ is an excellent node, that $N_y$ is a join node and that $y$ is in the bad side of $N_y$. Denote by $B$ the bad side of $N_y$ and $A$ the good side of $N_y$. By definition of excellent node, since $I$ has made all the good choices in the subtree of $A$, the maximum size of an independent set of $A$ is equal to $I \cap A$. So we have $\alpha(B) \leq \alpha(A) = |I \cap A|$. Since all the vertices of $N_y$ have the same neighborhood outside $N_y$, if $J'$ contains a vertex in $B$, then we can replace $B \cap J'$ by $I \cap N_y$ and not decrease the size of the independent set ($x$ can still be added to this independent set since vertices of $B_y$ have the same neighborhood). Note that this new independent set is in the same connected component of $TAR_k(G_i)$ as $I$: when we consider the sequence from $I$ to $J'$, we follow precisely the same sequence except that we do not delete the vertices of $I \cap A$ when they were deleted in the sequence and do not add vertices of $B$ when they are added. So we can assume that $J'$ is a subset of $I$. 


So we can assume that the set $J'$ is a subset of $I$. Nevertheless, since $x$ can be added to $J'$, it means that there exists a join node $B$ such that $x$ is in the good side and vertices of $I$ are in the bad side. By assumption, we know that $|I \cap B| \geq \alpha(B)$. So $J'$ has size at most $|I| - \alpha(B) < k$, a contradiction since $J'$ is in $TAR_k(G_i)$.
\end{proof}

Let us now assume that we are not in the first case, \emph{i.e.} there is no value of $\mathcal{L}$ in the interval between $k-1$ and $k-\alpha(B)$. The proof will consist in showing the following: first, if an independent set $I$ contains a vertex in $B$, then there exists in its connected component in $TAR_k(G_i)$ an independent set which does not contain any vertex in $B$. It means that $TAR_k(G_{i})$ is connected if $TAR_k(G_{i+1})$ is since we just have to find such independent sets in the components and find a reconfiguration sequence between them in $TAR(G_{i+1})$. Then we will show that if there is a sequence between two independent sets which do not contain a vertex of $B$ then there is sequence between them with no element containing a vertex of $B$. In other words, $TAR_k(G_{i+1})$ is connected if $TAR_k(G_i)$ is since a sequence between two stables sets of $TAR_k(G_{i+1})$ in $TAR_k(G_i)$ can be transformed into a reconfiguration between them in $TAR_k(G_{i+1})$.

Let $J$ be a maximal independent set of size at least $k$ containing a vertex of $B$. By assumption on $\mathcal{L}$, the size of $J \setminus B \geq k$ (it cannot be smaller than $k-\alpha(B)$ since otherwise the whole independent set $J$ would be smaller than $k$ since $J \cap B \leq  \alpha(B)$). So the size of $J$ is at least $k +\alpha(B)$. So if we delete the vertices of $B \cap J$, we still have an independent set of size at least $k$, which proves the first statement.

Now assume that there exists a sequence of additions and deletions between two independent sets $J$ and $J'$ which do not contain a vertex of $B$. Let $A$ be the good side of the join node whose bad side is $B$. Assume that at some step we add a vertex of $B$. Consider the first such addition.  Note that since we add this vertex, it means that no vertex of $A$ is in the independent set (vertices of the good and the bad sides are joined). Instead of adding this vertex, we can add a vertex of a maximum independent set in $A$. These two vertices have the same neighborhood outside $A \cup B$, so the constructed set is still a stable set. At each step where we add a vertex of $B$, we can add instead a vertex of a maximum stable set of $A$ (which is at least as large, by definition of good side) and each time we delete a vertex of the bad side, we can reflect this procedure on the good side. So we can build a sequence that does not rest upon vertices of $B$, simply by projecting all operations on vertices of $B$ onto the vertices of a maximum stable set in $A$.
\end{proof}

During the transformation of $G_i$ into $G_{i+1}$, we delete a join node, and there are at most $n-1$ of them, so our algorithm is indeed polynomial, and more explicitly runs in cubic time. This corresponds to $n$ times the execution time of Claim~\ref{claim:precomp}. 
\end{proof}

Note that, if we had $k$ at most the size of any maximal independent set, then our algorithm would run in quadratic time as there is no need to compute any list of values. Indeed, it suffices to check if the size of a minimum maximal stable set $I$ in $G \setminus (B \cup N(B))$ is at least $k + \alpha(B)$, and if so delete the corresponding join node. If not, then the union of $I$ and $B$ is also an independent set, which differs from the maximum independent set obtained by only taking good choices, and cannot be reconfigured into any other maximal independent set by definition of $B$. Since $|I \cup B|\geq k$ we have a counter-example to the connectivity of $TAR_k(G)$.

\section{Distance between independent sets}\label{sec:linear}

Let us now prove that it is possible to determine if two vertices are in the same connected component. The proof relies on some claims proved for Theorem~\ref{thm:cographs}.

\begin{theorem}
Given a cograph $G$, an integer $k$, and two independent sets $I, J$ of size at least $k$, it can be decided in linear time whether $I$ and $J$ are in the same connected component of $TAR_k(G)$.
\end{theorem}
\begin{proof}
Let us first describe the method. The precise algorithm which works in linear time will be detailed in Claims~\ref{claim:linear} and~\ref{claim:linear2}. We first compute a cotree decomposition of $G$. This can be done in linear time by~\cite{Corneil85}. We complete both $I$ and $J$ into two maximal independent sets (this completion can be done in linear time using the cotree). Each of them is associated with a stable-search as we have already observed in the proof of Theorem~\ref{thm:cographs}. We introduce a shortcut: when a maximal independent set $I$ induces the "good" maximum independent set in the subgraph induced by a node $N$ (i.e. it is obtained by making only good choices), we simply store this information as $I$ being good in $N$ (there is no ambiguity).

\begin{claim}\label{claim:linear}
The maximal independent set containing $I$ can be be computed in linear time. Moreover, this construction also gives the stable-search of the maximal stable set containing $I$.
\end{claim}
\begin{proof}
Let us first show, how, given the cotree $T$ of the cograph and an independent set, we can compute in linear time a maximal independent set which contains $I$. We first build the subtree of $T$ corresponding to $I$. We start bottom-up from the leaves. A leaf is selected if and only if it is in $I$. The father of a selected node is selected. Now we have the tree of $I$. Let us now complete it in order to obtain a maximal independent set. We just have to make a second search on this tree top-bottom from the root. If the node is an independent node, then both children are in the tree. If the node is a join node, then either one of the sons is already in the tree ($I$ intersects this subgraph) and we keep this node in the tree, or none of the sons are in the tree and then we put the good son in the tree. As in Claim~\ref{claim:nodeconstant}, a node which can be treated can be found in constant time, so this procedure takes a linear amount of time.

Now a maximal independent set containing $I$ is the independent set containing all the leaves of the constructed tree. Indeed, by construction, it contains the vertices of $I$. It is maximal since the tree is a stable-search: at any step, if we are on a join node, the tree contains vertices on exactly one side and if we are on a disjoint union node, it contains vertices of both sides. Moreover, the constructed tree is precisely the stable-search of this maximal independent set.
\end{proof}

For now we deal with each of $I$ and $J$ independently, say with $I$ here. Since $I$ is maximal by inclusion, the stable-search corresponding to $I$ is well-defined. Our algorithm just consists in a bottom-up operation. For every join node in the stable search, $I$ contains vertices either of the good or the bad side. We consider all such join nodes where $I$ contains vertices of the bad side, and sort them by increasing size of a maximum independent set in the bad size (this can be done in linear time as these are integer values between $1$ and $\alpha(G)$). If $I$ is a already the good maximum independent set and that list is empty, we stop here.

Let $N$ be the first join node in that order, \emph{i.e.} one that minimizes $\alpha(B)$, with $B$ being the bad side of $N$. Let $A$ be the good side of $N$. The independent set $I$ contains elements of $B$. If $|I\setminus B| \geq k$, we remove one by one all vertices in $I \cap B$, then add all vertices in the good maximum independent set in $A$, delete the node $N$ from the sorted bad vertices, and start again with the next join node in the order. Otherwise, we stop.

We claim that, when the algorithm stops, both $I$ and $J$ have been reconfigured into the ``highest'' maximal independent set possible. Then it suffices to check whether these two independent sets are the same. Indeed, if the two independent sets are the same then there is obviously a reconfiguration sequence from one into the other; so the output is yes. \\
Otherwise, let us prove that the answer is negative. By abuse of notations, let us still denote by $I$ and $J$ the obtained independent sets. Since $I \neq J$, w.l.o.g. there exists a join node $B$ of the cotree-decomposition, such that $I$ contains vertices of the good side of $B$ and $J$ contains vertices of the bad side of $B$ (since they are different they cannot have the same behavior on each join node of the cotree decomposition). According to the notations of the proof of Theorem~\ref{thm:cographs} and of Claim~\ref{claim:target}, since $J$ cannot be improved by the algorithm, it means that we have $|J| - \alpha(B') < k$ where $B'$ is the bad choice with the smallest independence number on $J$. Note that all the vertices of the good side of $B$ are target vertices for $J$. So Claim~\ref{claim:target} ensures that no independent set of the connected component of $J$ in $TAR_k(G)$ contains a vertex of the good side of $B$. But, by definition of $B$, the independent set $I$ intersects $B$. So the connected component of $J$ in $TAR_k(G)$ cannot contain $I$. \vspace{10pt}

Let us now show that this algorithm can be computed in linear time.

\begin{claim}\label{claim:linear2}
 The ``highest'' maximal independent set of $I$ can be computed in linear time.
\end{claim}
\begin{proof}
One of the key ideas in the algorithm consists in remembering only the size of the independent set which are in the cotree induced by a node instead of the independent set itself. Indeed, all the vertices which are below the node have the same neighborhood in the graph, so just computing its value is enough to know its behavior. Recall that Claim~\ref{claim:nodeconstant} ensures that we can compute an array containing the sizes of the maximum independent sets induced by each node in linear time.

First, create an array of size $n$. We start from the root of the tree and make a top-bottom search. For every join node where we make a bad choice, if the size of the maximal independent set of the bad size equals $k$, then put the node on the $k$-th entry of the array. Now the algorithm works as follows: initialize the size $s$ of the independent set as the size of $I$. For every entry $\ell$ in the increasing size of the array. If there is no element on the list, then move on to the next entry. Otherwise, let $B$ be an entry of the list. Either $s - \ell \geq k$, then delete $B$ of the list and increase $s$ by the difference between the size of the maximum independent set of the good side of $B$ minus the maximum independent set of the bad side. Indeed it corresponds to the case where we can delete all the vertices of the bad side and replace them by vertices of the good one. The size of the independent set increases by this difference.
Now if $s - \ell<k$, then it means that we cannot eliminate all the vertices of the bad side to replace them by vertices of the good side (and since we treat bad sides in the increasing order it means that no other bad side can be treated). 

Note that all these operations for a given node can be done in constant time. Each node appears in at most one list and the array is linear. So the execution time of this algorithm is linear.

Note that this algorithm does precisely what we described before. The unique difference is that it just keeps in mind the sizes of the independent sets and not the independent sets themselves since, as long as there is no bad choice in a subtree, we can easily find this independent set. If there are two bad choices where one of them is an ancestor of another, then the maximum independent set of the ancestor is at least as large as that of the small one. Since we treat bad nodes in the increasing order (of their maximum size), if it is strictly smaller, the small one is treated before. If they are equal, since we search top-bottom the tree, the bad choice which is closest from the leaves will be treated first. So at every step, we treat a node which is closest from the leaves and which has a smallest maximum independent set.
\end{proof}

Note, finally, that two independent sets can be compared in linear time, since we just have to verify if the stable-searches (which are words on at most $n-1$ letters) are the same or not.
\end{proof}

\bibliographystyle{plain}

\bibliography{biblio}

\begin{thebibliography}{10}

\bibitem{BonamyB14}
M.~Bonamy and N.~Bousquet.
\newblock Recoloring graphs via tree decompositions.
\newblock Arxiv preprint, 2014.
\newblock http://arxiv.org/abs/1403.6386.

\bibitem{b14}
P.~Bonsma.
\newblock Independent set reconfiguration in cographs.
\newblock to appear in {WG'}14, 2014.
\newblock available at http://arxiv.org/abs/1402.1587.

\bibitem{BonsmaKW14}
P.~Bonsma, M.~Kaminski, and M.~Wrochna.
\newblock Reconfiguring independent sets in claw-free graphs.
\newblock Arxiv preprint, 2014.
\newblock http://arxiv.org/abs/1403.0359.

\bibitem{Corneil85}
D.~Corneil, Y.~Perl, and L.~Stewart.
\newblock A linear recognition algorithm for cographs.
\newblock {\em SIAM Journal on Computing}, 14(4):926--934, 1985.

\bibitem{Diestel2005}
R.~Diestel.
\newblock {\em Graph Theory}, volume 173 of {\em Graduate Texts in
  Mathematics}.
\newblock Springer-Verlag, Heidelberg, third edition, 2005.

\bibitem{Gopalan09}
P.~Gopalan, P.~Kolaitis, E.~Maneva, and C.~Papadimitriou.
\newblock The connectivity of boolean satisfiability: Computational and
  structural dichotomies.
\newblock {\em SIAM J. Comput.}, pages 2330--2355, 2009.

\bibitem{HearnD05}
R.~Hearn and E.~Demaine.
\newblock Pspace-completeness of sliding-block puzzles and other problems
  through the nondeterministic constraint logic model of computation.
\newblock {\em Theor. Comput. Sci.}, 343(1-2):72--96, October 2005.

\bibitem{Ito2011}
T.~Ito, E.~Demaine, N.~Harvey, C.~Papadimitriou, M.~Sideri, R.~Uehara, and
  Y.~Uno.
\newblock On the complexity of reconfiguration problems.
\newblock {\em Theor. Comput. Sci.}, 412(12-14):1054--1065, 2011.

\bibitem{Ito2009}
T.~Ito, M.~Kamiński, and E.~Demaine.
\newblock Reconfiguration of list edge-colorings in a graph.
\newblock In {\em Alg. \& Data Struct.}, volume 5664 of {\em Lecture Notes in
  Computer Science}, pages 375--386. 2009.

\bibitem{KaminskiMM12}
M.~Kamiński, P.~Medvedev, and M.~Milanič.
\newblock Complexity of independent set reconfigurability problems.
\newblock {\em Theoretical Computer Science}, 439(0):9 -- 15, 2012.

\bibitem{lerchs71}
H.~Lerchs.
\newblock On cliques and kernels.
\newblock Technical report, 1971.

\end{thebibliography}

\end{document}